\documentclass[letterpaper, 10 pt, conference]{ieeeconf}%
\usepackage{bm}
\usepackage{amssymb}
\usepackage{epsfig}
\usepackage{amsmath}
\usepackage{xcolor}
\usepackage{graphicx}
\usepackage{epstopdf}
\usepackage{amsfonts}%
\usepackage{indentfirst}
\newtheorem{problem}{\textbf{Problem}}
\newtheorem{definition}{\textbf{Definition}}

\newtheorem{theorem}{\rm\textbf{Theorem}}
\newtheorem{lemma}{\rm\textbf{Lemma}}

\newtheorem{remark}{\rm\textbf{Remark}}
\providecommand{\U}[1]{\protect\rule{.1in}{.1in}}
\IEEEoverridecommandlockouts
\begin{document}

\title{{\LARGE \textbf{Control Barrier Functions for Systems with High Relative Degree}}}
\author{Wei Xiao and Calin Belta\thanks{This work was supported in
part by the NSF under grants IIS-1723995 and CPS-1446151.}\thanks{The authors are
with the Division of Systems Engineering and Center for Information and
Systems Engineering, Boston University, Brookline, MA, 02446, USA
\texttt{{\small \{xiaowei,cbelta\}@bu.edu}}}}
\maketitle

\begin{abstract}
This paper extends control barrier functions (CBFs) to high order control barrier functions (HOCBFs) that can be used for high relative degree constraints. The proposed HOCBFs are more general than recently proposed (exponential) HOCBFs. 
We introduce high order barrier functions (HOBF), and show that their satisfaction of Lyapunov-like conditions implies the forward invariance of the intersection of a series of sets. We then introduce HOCBF, and show that any control input that satisfies the HOCBF constraints renders the intersection of a series of sets forward invariant. We formulate optimal control problems with constraints given by HOCBF and control Lyapunov functions (CLF) and analyze the influence of the choice of the class $\mathcal{K}$ functions used in the definition of the HOCBF on the size of the feasible control region. We also provide a promising method to address the conflict between HOCBF constraints and control limitations by penalizing the class $\mathcal{K}$ functions.
We illustrate the proposed method on an adaptive cruise control problem.  
\end{abstract}

\thispagestyle{empty} \pagestyle{empty}


\section{INTRODUCTION}
\label{sec:intro}

 Barrier functions (BF) are 
 Lyapunov-like functions \cite{Tee2009}\cite{Wieland2007}, whose use can be traced back to optimization problems \cite{Boyd2004}. More recently, they have been employed in verification and control, e.g., to prove set invariance \cite{Aubin2009}\cite{Prajna2004}\cite{Prajna2007}\cite{Wisniewski2013} and for multi-objective control \cite{Panagou2013}. Control BF (CBF) are extensions of BFs for control systems. Recently, it has been shown that CBF can be combined with
 control Lyapunov functions (CLF) \cite{Sontag1983}\cite{Artstein1983}\cite{Freeman1996}\cite{Aaron2012} as constraints to form quadratic programs (QP) \cite{Galloway2013} that are solved in real time. The CLF constraints can be relaxed \cite{Aaron2014} such that they do not conflict with the CBF constraints to form feasible QPs.
 
 In \cite{Tee2009} it was proved that if a barrier function for a given set satisfies Lyapunov-like conditions, then the set is forward invariant. A less restrictive form of a barrier function, which is allowed to grow when far away from the boundary of the set, was proposed in \cite{Aaron2014}. 
 Another approach that allows a barrier function to be zero was proposed in \cite{Glotfelter2017} \cite{Lindemann2018}. This simpler form has also been considered in time-varying cases and applied to enforce Signal Temporal Logic (STL) formulas as hard constraints \cite{Lindemann2018}.
 
 The barrier functions from \cite{Aaron2014} and \cite{Glotfelter2017} work for constraints that have relative degree one (with respect to the dynamics of the system). A backstepping approach was introduced in \cite{Hsu2015} to address higher relative degree constraints, and it was shown to work for relative degree two. A CBF method for position-based constraints with relative degree two was also proposed in \cite{Wu2015}. A more general form, which works for 
 for arbitrarily high relative degree constraints, was proposed in \cite{Nguyen2016}. The method in \cite{Nguyen2016} employs input-output linearization and finds a pole placement controller with negative poles to stabilize the barrier function to zero. Thus, this barrier function is an exponential barrier function.

In this paper, we propose a barrier function for high relative degree constraints, called high-order control barrier function (HOCBF), which is simpler and more general than the one from \cite{Nguyen2016}. 
Our barrier functions are not 
restricted to exponential functions, and are determined  by a set of class $\mathcal{K}$ functions. The general form of a barrier function proposed here is associated with the forward invariance of the intersection of a series of sets. 

We formulate optimal control problems with constraints given by HOCBF and CLF and analyze the influence of the choice of the class $\mathcal{K}$ functions used in the definition of the HOCBF on the size of the feasible control region and on the performance of the system. 
We also show that, by applying penalties on the class $\mathcal{K}$ functions, we can manage possible conflicts between HOCBF constraints and other constraints, such as control limitations. The main advantage of using the general form of HOCBF proposed in this paper is that it can be adapted to different types of systems and constraints. 

We illustrate the proposed method on an adaptive cruise control problem. We consider 
square root, linear and quadratic class $\mathcal{K}$ functions in the HOCBF. The simulations show that the 
 results are heavily dependent on the choice of the class $\mathcal{K}$ functions.

 
 \section{PRELIMINARIES}
 \label{sec:prelim}

 \begin{definition} \label{def:classk}
 	({\it Class $\mathcal{K}$ function} \cite{Khalil2002}) A continuous function $\alpha:[0,a)\rightarrow[0,\infty), a > 0$ is said to belong to class $\mathcal{K}$ if it is strictly increasing and $\alpha(0)=0$. 
 \end{definition}

\begin{lemma}  \label{lem:bf} (\cite{Glotfelter2017})
	 Let $b: [t_0, t_1]\rightarrow \mathbb{R}$ be a continuously differentiable function.  If $\dot b(t) \geq \alpha(b(t))$, for all $t\in[t_0,t_1]$, where $\alpha$ is a class $\mathcal{K}$ function of its argument, and $b(t_0)\geq 0$, then $b(t)\geq 0, \forall t\in[t_0,t_1]$.
\end{lemma}

\vspace{2mm}
 Consider a system of the form
 \begin{equation}\label{eq:sys}%
 \dot {\bm{x}} = f(\bm x),
 \end{equation}
 with $\bm x\in \mathbb{R}^n$ and $f:\mathbb{R}^n\rightarrow \mathbb{R}^n$ locally Lipschitz. Solutions $\bm x(t)$ of (\ref{eq:sys}), starting at $\bm x(t_0)$, $t\geq t_0$, are forward complete.
 
 \vspace{2mm} 
In this paper, we also consider affine control systems in the form
 \begin{equation} \label{eq:affine}%
 \dot {\bm{x}} = f(\bm x) + g(\bm x)\bm u,
 \end{equation}
 where  $\bm x\in \mathbb{R}^n$, $f$ is as defined above, $g:\mathbb{R}^n \rightarrow \mathbb{R}^{n\times q}$ is locally Lipschitz, and $\bm u\in U \subset \mathbb{R}^q$ ($U$ denotes the control constraint set) is Lipschitz continuous. Solutions $\bm x(t)$ of (\ref{eq:affine}), starting at $\bm x(t_0)$, $t\geq t_0$, are forward complete.

 \begin{definition} \label{def:forwardinv}
 	A set $C\subset\mathbb{R}^n$ is forward invariant for system (\ref{eq:sys}) (or (\ref{eq:affine})) if its solutions starting at all $\bm x(t_0) \in C$ satisfy $\bm x(t)\in C$ for $\forall t\geq t_0$.
 \end{definition} 

Let 
	\begin{equation}\label{eqn:C-b}
	C := \{\bm x\in \mathbb{R}^n:b(\bm x)\geq 0\},
	\end{equation} 
where $b:\mathbb{R}^n\rightarrow \mathbb{R}$ is a continuously differentiable function. 

\begin{definition} \label{def:bf2}
	({\it Barrier function} \cite{Glotfelter2017} \cite{Lindemann2018}): 
The continuously differentiable function $b:\mathbb{R}^n\rightarrow \mathbb{R}$ is a barrier function (BF) for system (\ref{eq:sys}) if there exists a class $\mathcal{K}$ function $\alpha$ such that
	\begin{equation}
	\dot b(\bm x) + \alpha(b(\bm x)) \geq 0,
	\end{equation}
\end{definition}
for all $\bm x\in C$.

\begin{theorem} \label{thm:bf2} (\cite{Lindemann2018})
	Given a set $C$ as in Eqn. (\ref{eqn:C-b}), if there exist a BF $b:C\rightarrow \mathbb{R}$, then $C$ is forward invariant for system (\ref{eq:sys}).
\end{theorem}

  \begin{definition} \label{def:cbf2}
  	({\it Control barrier function} \cite{Glotfelter2017} \cite{Lindemann2018}): Given a set $C$ as in Eqn. (\ref{eqn:C-b}), $b(\bm x)$ is a control barrier function (CBF) for system (\ref{eq:affine}) if there exists a class $\mathcal{K}$ function $\alpha$ such that
 	\begin{equation} \label{eq:cbf2}
 	L_fb(\bm x)+L_gb(\bm x) \bm u + \alpha(b(\bm x)) \geq 0 
 	\end{equation}
 \end{definition}
for all $\bm x\in C$.
 
 \begin{theorem} \label{thm:cbf2}(\cite{Glotfelter2017}, \cite{Lindemann2018})
 	 Given a CBF $b$ with the associated set $C$ from Eqn. (\ref{eqn:C-b}), any Lipschitz continuous controller $ \bm u \in K_{cbf}(\bm x)$, with 
 	$$K_{cbf}(\bm x) := \{ \bm u\in U: L_fb(\bm x)+L_gb(\bm x) \bm u + \alpha(b(\bm x)) \geq 0 \},$$
 	renders the set $C$ forward invariant for affine control system (\ref{eq:affine}). 
 \end{theorem}

 \begin{remark}
	The barrier functions in Defs. \ref{def:bf2} and \ref{def:cbf2} can be seen as more general forms of the ones defined in \cite{Aaron2014}. In particular, it can be shown that, for each barrier function defined in \cite{Aaron2014}, we can always find one in the form in Def. \ref{def:bf2}.
\end{remark}

 \begin{definition}  \label{def:clf}
 	({\it Control Lyapunov function} \cite{Aaron2012}) A continuously differentiable function $V: \mathbb{R}^n\rightarrow \mathbb{R}$ is a globally and exponentially stabilizing control Lyapunov function (CLF) for system (\ref{eq:affine}) if there exist constants $c_1 >0, c_2>0, c_3>0$ such that
 	\begin{equation}
 	c_1||\bm x||^2 \leq V(\bm x) \leq c_2 ||\bm x||^2
 	\end{equation}
 	\begin{equation}\label{eqn:clf}
 	\underset{u\in U}{inf} \lbrack L_fV(\bm x)+L_gV(\bm x) \bm u + c_3V(\bm x)\rbrack \leq 0.
 	\end{equation}
 	for $\forall \bm x\in \mathbb{R}^n$.
 \end{definition}
 
 \begin{theorem} \label{thm:clf}
 Given an exponentially stabilizing CLF $V$ as in Def. \ref{def:clf}, any Lipschitz continuous controller $ \bm u \in K_{clf}(\bm x)$, with
 $$K_{clf}(\bm x) := \{\bm u\in U: L_fV(\bm x)+L_gV(\bm x) \bm u + c_3V(\bm x) \leq 0\},$$
 exponentially stabilizes system (\ref{eq:affine}) to its zero dynamics \cite{Aaron2012} (defined by the dynamics of the internal part if we transform the system to standard form and set the output to zero \cite{Khalil2002}).
 \end{theorem}

\begin{definition} \label{def:relative}
	({\it Relative degree})
	The relative degree of a continuously differentiable function $b:\mathbb{R}^n\rightarrow \mathbb{R}$ with respect to system (\ref{eq:affine}) is the number of times we need to differentiate it along the dynamics of (\ref{eq:affine}) until the control $\bm u$ explicitly shows. 
\end{definition}

In this paper, since function $b$ is used to define a constraint $b(\bm x)\geq 0$, we will also refer to the relative degree of $b$ as the relative degree of the constraint. 

Many existing works \cite{Aaron2014}, \cite{Lindemann2018}, \cite{Nguyen2016} combine CBF and CLF with quadratic costs to form optimization problems. Time is discretized and an optimization problem with constraints given by CBF and CLF is solved at each time step. Note that these constraints are linear in control since the state is 
fixed at the value at the beginning of the interval, and therefore the optimization problem is a quadratic program (QP). The optimal control obtained by solving the QP is applied 
at the current time step and held constant for the whole interval. The dynamics (\ref{eq:affine}) is updated, and the procedure is repeated. It is important to note that this method works conditioned upon the fact that the control input shows up in (\ref{eq:cbf2}), i.e., $L_gb(\bm x)\ne 0$.

\section{HIGH ORDER CONTROL BARRIER FUNCTIONS}
\label{sec:hocbf}

In this section, we define high order barrier functions (HOBF) and high order control barrier functions (HOCBF). We use a simple example to motivate the need for such functions and to illustrate the main ideas. 

\subsection{Example: Simplified Adaptive Cruise Control}
\label{sec:hocbf:moti}

Consider the simplified adaptive cruise control (SACC) problem\footnote{A more realistic version of this problem, called the adaptive cruise control problem (ACC), is defined in Sec.  \ref{sec:ACC}.} with the vehicle dynamics for vehicle $i\in S(t)$ (where $S(t)$ denotes the set of indices of vehicles in an urban area at time $t$) in the form:
\begin{equation} \label{eqn:simpledynamics}
\left[\begin{array}{c} 
\dot x_i(t)\\
\dot v_i(t)
\end{array} \right]=
\left[\begin{array}{c}  
v_i(t)\\
0
\end{array} \right] + 
\left[\begin{array}{c}  
0\\
1
\end{array} \right]u_i(t),
\end{equation}
where $x_i(t)$ and $v_i(t)$ denote the position and velocity of vehicle $i$ along its lane, respectively, and $u_i(t)$ is its control input.

Following \cite{Malikopoulos2018}, we require that the distance between vehicle $i\in S(t)$ and its immediately preceding vehicle $i_p\in S(t)$ 
(the coordinates $x_i$ and $x_{i_p}$ of vehicles $i$ and $i_p$, respectively, are measured from the same origin and $x_{i_p}(t) \geq x_i(t), \forall t\geq t_i^0$) 
be greater than a constant $\delta > 0$ for all the times, i.e.,
\begin{equation} \label{eqn:safety}
x_{i_p}(t) - x_i(t) \geq \delta,\forall t\geq t_i^0.
\end{equation}

Assume $i_p$ runs at constant speed $v_0$. In order to use CBF to find control input for $i$ such that the safety contraint (\ref{eqn:safety}) is always satisified, any control input $u_i(t)$ should satisfy
\begin{equation}\label{eqn:safety_ex2}
\begin{aligned}
\underbrace{v_0 - v_i(t)}_{L_fb(\bm x_i(t))} + \underbrace{0}_{L_gb(\bm x_i(t))}\times u_i(t) + \underbrace{x_{i_p}(t) - x_i(t) - \delta}_{\alpha(b(\bm x_i(t)))}\geq 0.
\end{aligned}
\end{equation}

Notice that $L_gb(\bm x_i(t)) = 0$ in (\ref{eqn:safety_ex2}), so the control input $u_i(t)$ does not show up. 
Therefore, we cannot use these barrier functions to formulate an optimization problem as described at the end of Sec. \ref{sec:prelim}. 

\subsection{High Order Barrier Function (HOBF)}
\label{sec:hobf}

As in \cite{Lindemann2018}, we consider a time-varying function to define an invariant set for system (\ref{eq:sys}). For a $m^{th}$ order differentiable function $b: \mathbb{R}^n \times [t_0, \infty) \rightarrow \mathbb{R}$ (where $t_0$ denotes the initial time), we define a series of functions $\psi_0: \mathbb{R}^n \times [t_0, \infty) \rightarrow \mathbb{R}, \psi_1: \mathbb{R}^n \times [t_0, \infty) \rightarrow \mathbb{R}, \psi_2: \mathbb{R}^n \times [t_0, \infty) \rightarrow \mathbb{R},\dots, \psi_{m}: \mathbb{R}^n \times [t_0, \infty) \rightarrow \mathbb{R}$ in the form:
\begin{equation} \label{eqn:functions}
\begin{aligned}
\psi_0(\bm x,t) :=& b(\bm x,t)\\
\psi_1(\bm x,t) :=& \dot \psi_0(\bm x,t) + \alpha_1(\psi_0(\bm x,t)),\\
\psi_2(\bm x,t) := &\dot \psi_1(\bm x,t) + \alpha_2(\psi_1(\bm x,t)),\\
&\vdots\\
\psi_m(\bm x,t) :=& \dot \psi_{m-1}(\bm x,t) + \alpha_m(\psi_{m-1}(\bm x,t)),
\end{aligned}
\end{equation}
where $\alpha_1(.),\alpha_2(.),\dots, \alpha_{m}(.)$ denote class $\mathcal{K}$ functions of their argument.

We further define a series of sets $C_1(t), C_2(t),\dots, C_{m}(t)$ associated with (\ref{eqn:functions}) in the form:
\begin{equation} \label{eqn:sets}
\begin{aligned}
C_1(t) := &\{\bm x \in \mathbb{R}^n: \psi_0(\bm x,t) \geq 0\}\\
C_2(t) := &\{\bm x \in \mathbb{R}^n: \psi_1(\bm x,t)  \geq 0\}\\
&\vdots\\
C_{m}(t) := &\{\bm x \in \mathbb{R}^n: \psi_{m-1}(\bm x,t) \geq 0\}
\end{aligned}
\end{equation}

\vspace{2ex}
\begin{definition} \label{def:hobf}
	Let $C_1(t), C_2(t),\dots, C_{m}(t)$ be defined by (\ref{eqn:sets}) and $\psi_1(\bm x,t), \psi_2(\bm x,t),\dots, \psi_{m}(\bm x,t)$ be defined by (\ref{eqn:functions}). A function $b: \mathbb{R}^n\times[t_0, \infty)\rightarrow \mathbb{R}$ is a high order barrier function (HOBF) that is $m^{th}$ order differentiable for system (\ref{eq:sys}) if there exist differentiable class $\mathcal{K}$ functions $\alpha_1,\alpha_2\dots \alpha_{m}$ such that
\begin{equation}\label{eqn:hocbf}
\psi_{m}(\bm x(t),t) \geq 0
\end{equation}
for all $(\bm x,t)\in C_1(t) \cap C_2(t)\cap,\dots, \cap C_{m}(t) \times [t_0, \infty)$.
\end{definition}

Note that $\dot \psi_{i}(\bm x,t) = \frac{d\psi_{i}(\bm x,t)}{dt} = \frac{\partial\psi_{i}(\bm x,t)}{\partial \bm x}\dot {\bm x}  + \frac{\partial\psi_{i}(\bm x,t)}{\partial t}$, $\forall i\in\{1,2,\dots,m-1\}$.

\vspace{2ex}
\begin{theorem} \label{thm:hobf}
	The set $C_1(t)\cap C_2(t)\cap,\dots, \cap C_{m}(t)$ is forward invariant for system (\ref{eq:sys}) if $b(\bm x(t),t)$ is a HOBF that is $m^{th}$ order differentiable.
\end{theorem}

\begin{proof}
 If $b(\bm x(t),t)$ is a HOBF that is $m^{th}$ order differentiable, then $\psi_{m}(\bm x(t),t) \geq 0$ for $\forall t\in[t_0,\infty)$, i.e., $\dot \psi_{m-1}(\bm x(t),t) + \alpha_m(\psi_{m-1}(\bm x(t),t)) \geq 0$. By Lemma \ref{lem:bf}, since $\bm x(t_0) \in C_{m}(t_0)$ (i.e., $\psi_{m-1}(\bm x(t_0),t_0))\geq 0$, and $\psi_{m-1}(\bm x(t),t)$ is an explicit form of $\psi_{m-1}(t)$), then $\psi_{m-1}(\bm x(t),t)) \geq 0$, $\forall t\in[t_0,\infty)$, i.e., $\dot \psi_{m-2}(\bm x(t),t) + \alpha_{m-1}(\psi_{m-2}(\bm x(t),t)) \geq 0$. Again, by Lemma 1, since $\bm x(t_0) \in C_{m-1}(t_0)$, we also have $\psi_{m-2}(\bm x(t),t)) \geq 0$, $\forall t\in[t_0,\infty)$. Iteratively, we can get $\bm x(t) \in C_i(t)$, $\forall i \in\{1,2,\dots,m\}, \forall t\in[t_0,\infty)$. Therefore, the sets $C_1(t), C_2(t)\dots C_{m}(t)$ are forward invariant.
\end{proof}

\begin{remark} \label{rem:intersect}
	 The sets $C_1(t), C_2(t),\dots, C_{m}(t)$ should have a non-empty intersection at $t_i^0$ in order to satisfy the forward invariance condition starting from $t_i^0$ in Thm. \ref{thm:hobf}. If $b(\bm x(t_0), t_0)\geq 0$, we can always choose proper class $\mathcal{K}$ functions $\alpha_1(.),\alpha_2(.),\dots, \alpha_{m}(.)$ to make $\psi_1(\bm x(t_0),t_0) \geq 0, \psi_2(\bm x(t_0),t_0)\geq 0,\dots, \psi_{m-1}(\bm x(t_0),t_0)\geq 0$. There are some extreme cases, however, when this is not possible. For example, if $\psi_0(\bm x(t_0), t_0) = 0$ and $\dot \psi_0(\bm x(t_0), t_0) < 0$, then $\psi_1(\bm x(t_0), t_0)$ is always negative no matter how we choose $\alpha_1(\cdot)$. Similarly, if $\psi_0(\bm x(t_0), t_0) = 0$, $\dot \psi_0(\bm x(t_0), t_0) = 0$ and $\dot \psi_1(\bm x(t_0), t_0) < 0$, $\psi_2(\bm x(t_0), t_0)$ is also always negative, etc.. To deal with such extreme cases (as with the case when $b(\bm x(t_0),t_0) < 0$), we would need a feasibility enforcement method, which is beyond the scope of this paper.
\end{remark}

\subsection{High Order Control Barrier Function (HOCBF)}
\label{sec:hocbf_sub}

\vspace{2ex}
\begin{definition} \label{def:hocbf}
	Let $C_1(t), C_2(t),\dots, C_{m}(t)$ be defined by (\ref{eqn:sets}) and $\psi_1(\bm x,t), \psi_2(\bm x,t),\dots, \psi_{m}(\bm x,t)$ be defined by (\ref{eqn:functions}). A function $b: \mathbb{R}^n\times[t_0, \infty)\rightarrow \mathbb{R}$ is a high order control barrier function (HOCBF) of relative degree $m$ for system (\ref{eq:affine}) if there exist differentiable class $\mathcal{K}$ functions $\alpha_1,\alpha_2,\dots, \alpha_{m}$ such that
\begin{equation}\label{eqn:constraint}
\begin{aligned}
L_f^{m}b(\bm x,t) + L_gL_f^{m-1}b(\bm x,t)\bm u + \frac{\partial^{m}b(\bm x,t)}{\partial t^m}\\+ O(b(\bm x,t)) + \alpha_m(\psi_{m-1}(\bm x,t)) \geq 0,
\end{aligned}
\end{equation}
for all $(\bm x,t)\in C_1(t) \cap C_2(t)\cap,\dots, \cap C_{m}(t) \times [t_0, \infty)$. In the above equation, $O(.)$ denotes the remaining Lie derivatives along $f$ and partial derivatives with respect to $t$ with degree less than or equal to $m-1$.
\end{definition}

Given a HOCBF $b$, we define the set of all control values that satisfy (\ref{eqn:constraint}) as:
\begin{equation}
\begin{aligned}
K_{hocbf} = \{\bm u\in U: L_f^{m}b(\bm x,t) + L_gL_f^{m-1}b(\bm x,t)\bm u \\+ \frac{\partial^{m}b(\bm x,t)}{\partial t^m} + O(b(\bm x,t))  + \alpha_m(\psi_{m-1}(\bm x,t)) \geq 0\}
\end{aligned}
\end{equation}

\vspace{2ex}
\begin{theorem} \label{cor:hocbf}
	 Given a HOCBF $b(\bm x, t)$ from Def. \ref{def:hocbf} with the associated sets $C_1(t), C_2(t),\dots, C_{m}(t)$ defined by (\ref{eqn:sets}), if $\bm x(t_0) \in C_1(t_0) \cap C_2(t_0)\cap,\dots,\cap C_{m}(t_0)$, then any Lipschitz continuous controller $\bm u(t)\in K_{hocbf}$ renders the set
	 $C_1(t)\cap C_2(t)\cap,\dots, \cap C_{m}(t)$ forward invariant for system (\ref{eq:affine}).
\end{theorem}

\begin{proof}
	Since $\bm u(t)$ is Lipschitz continuous and $\bm u(t)$ only shows up in the last equation of (\ref{eqn:functions}) when we take Lie derivative on (\ref{eqn:functions}), we have that $\psi_{m}(\bm x, t)$ is also Lipschitz continuous. The system states in (\ref{eq:affine}) are all continuously differentiable, so $\psi_{1}(\bm x, t), \psi_{2}(\bm x, t), \dots, \psi_{m-1}(\bm x, t)$ are also continuously differentiable. Therefore, the HOCBF has the same property as the HOBF in Def. \ref{def:hobf}, and the proof is the same as the one for Theorem \ref{thm:hobf}.
	\end{proof}

Note that, if we have a constraint $b(\bm x, t)\geq 0$ with relative degree $m$, then the number of sets is also $m$.

\vspace{2ex}
\begin{remark} \label{rem:var}
The general, time-varying HOCBF 
introduced in Def. \ref{def:hocbf}, can be used for general, time-varying 
constraints (e.g., signal temporal logic specifications \cite{Lindemann2018}) and systems.  However, the ACC problem (we will consider in this paper) has time-invariant system dynamics and constraints. Therefore, in the rest of this paper, we focus on time-invariant versions for simplicity. 
\end{remark}

\begin{remark} \label{rem:relation}
	({\it Relationship between time-invariant HOCBF and exponential CBF in }\cite{Nguyen2016}) In Def. \ref{def:hobf}, if we set class $\mathcal{K}$ functions $\alpha_1,\alpha_2\dots \alpha_{m}$ to be linear functions with positive coefficients, then we can get exactly the same formulation as in \cite{Nguyen2016} that is obtained through input-output linearization. i.e., 
\begin{equation} \label{eqn:functionslinear}
\begin{aligned}
\psi_1(\bm x) :=& \dot b(\bm x) + k_1b(\bm x)\\
\psi_2(\bm x) := &\dot \psi_1(\bm x) + k_2\psi_1(\bm x)\\
&\vdots\\
\psi_m(\bm x) :=& \dot \psi_{m-1}(\bm x) + k_m\psi_{m-1}(\bm x)
\end{aligned}
\end{equation}
where $k_1 > 0, k_2 > 0,\dots,k_m > 0$. The time-invariant HOCBF is the generalization of exponential CBF.
\end{remark}

\vspace{2ex}
$\textbf{Example revisited.}$ 
For the SACC problem introduced in Sec.\ref{sec:hocbf:moti}, the relative degree of the constraint from Eqn. (\ref{eqn:safety}) is 2. Therefore, we need a HOCBF with $m = 2$.

We choose quadratic class $\mathcal{K}$ functions for both $\alpha_1(\cdot)$ and $\alpha_2(\cdot)$, i.e.,  $\alpha_1(b(\bm x_i(t))) = b^2(\bm x_i(t))$ and $\alpha_2(\psi_1(\bm x_i(t))) = \psi_1^2(\bm x_i(t))$. In order for $b(\bm x_i(t)) := x_{i_p}(t) - x_i(t) - \delta$ to be a HOCBF for (\ref{eqn:simpledynamics}), it should satisfy the following constraint:
\begin{equation}
\begin{aligned}
\ddot b(\bm x_i(t)) + 2\dot b(\bm x_i(t))b(\bm x_i(t)) + \dot b^2(\bm x_i(t)) \\+ 2\dot b(\bm x_i(t)) b^2(\bm x_i(t)) + b^4(\bm x_i(t)) \geq 0.
\end{aligned}
\end{equation}

A control input $u(t)$ should satisfy
\begin{equation} \label{eqn:safety_ex3}
\begin{aligned}
L_f^2 b(\bm x_i(t)) + L_gL_f b(\bm x_i(t))u_i(t) + 2b(\bm x_i(t))L_f b(\bm x_i(t))\\ + (L_f b(\bm x_i(t)))^2 + 2b^2(\bm x_i(t))L_f b(\bm x_i(t)) + b^4(\bm x_i(t)) \geq 0.
\end{aligned}
\end{equation}

Note that $L_gL_f b(\bm x_i(t),t)\ne 0$ in (\ref{eqn:safety_ex3}) and the initial conditions are $ b(\bm x_i(t_i^0)) \geq 0$ and $\dot b(\bm x_i(t_i^0)) + b^2(\bm x_i(t_i^0)) \geq 0$.

\subsection{Optimal Control for Time-Invariant Constraints}
\label{sec:oc}

 Consider an optimal control problem for system (\ref{eq:affine}) with the cost defined as:
\begin{equation}\label{eqn:cost}
J(\bm u(t)) = \int_{t_0}^{t_f}\mathcal{C}(||\bm u(t)||)dt
\end{equation}
where $||\cdot||$ denotes the 2-norm of a vector. $t_0,t_f$ denote the initial and final times, respectively, and  $\mathcal{C}(\cdot)$ is a strictly increasing function of its argument. Assume a time-invariant (safety) constraint $b(\bm x) \geq 0$ with relative degree $m$ has to be satisfied by system (\ref{eq:affine}). Then the control input $\bm u$ should satisfy the time-invariant HOCBF version of the constraint from (\ref{eqn:constraint})):
\begin{equation}\label{eqn:constraintOC}
\begin{aligned}
L_f^{m}b(\bm x) + L_gL_f^{m-1}b(\bm x)\bm u + O(b(\bm x)) + \alpha_m(\psi_{m-1}(\bm x)) \geq 0
\end{aligned}
\end{equation}
with $\bm x(t_0)\in C_1(t_0)\cap C_2(t_0)\cap,\dots, \cap C_{m}(t_0)$.

If convergence to a given state is required in addition to optimality and safety, then, as 
in \cite{Aaron2014}, HOCBF can be combined with CLF. We discretize the time and formulate a cost (\ref{eqn:cost}) while subjecting to the HOCBF constraint (\ref{eqn:constraintOC}) and CLF constraint (\ref{eqn:clf}) at each time. With the optimal control input $\bm u$ obtained from (\ref{eqn:cost}) subject to (\ref{eqn:constraintOC}), (\ref{eqn:clf}) and (\ref{eq:affine}) at each time instant, we update the system dynamics (\ref{eq:affine}) for each time step, and the procedure is repeated. Then $C_1(t)\cap C_2(t)\cap,\dots, \cap C_m(t)$ is forward invariant, i.e., the safety constraint $b(\bm x) \geq 0$ is satisfied for (\ref{eq:affine}), $\forall t\in[t_0, t_f]$.

\subsection{Time-invariant HOCBF Properties}
\label{sec:property}
In this section, we consider how we should properly choose class $\mathcal{K}$ functions $\alpha_1, \alpha_2, \dots, \alpha_m$ for a time-invariant HOCBF such that the performance of system (\ref{eq:affine}) and the feasibility of the optimal control problem defined in Sec. \ref{sec:oc} are improved. For simplicity, in this section we assume that the term $L_gL_f^{m-1}b(\bm x(t))$ in (\ref{eqn:constraintOC}) does not change sign for all $t\in[t_0, t_f]$.

\subsubsection{\textbf{Feasible Region of Control Input}} For an optimal control problem as defined in Sec. \ref{sec:oc}, we rewrite the time-invariant HOCBF constraint (\ref{eqn:constraintOC}) as
\begin{equation}\label{eqn:hocbfcontrol}
 \bm u  \leq \frac{ L_f^{m}b(\bm x) + O(b(\bm x)) + \alpha_m(\psi_{m-1}(\bm x))}{-L_gL_f^{m-1}b(\bm x)}
\end{equation}
if $L_gL_f^{m-1}b(\bm x) < 0$ (otherwise, the inequality in (\ref{eqn:hocbfcontrol}) should be $\geq$).

Suppose we also have control limitations 
\begin{equation} \label{eqn:controllimit}
\bm u_{min}\leq\bm u(t) \leq \bm u_{max}, \forall t\in [t_0,t_f]
\end{equation}
for system (\ref{eq:affine}), where $\bm u_{min},\bm u_{max}\in \mathbb{R}^q$. If $\frac{ L_f^{m}b(\bm x(t_0)) + O(b(\bm x(t_0))) + \alpha_m(\psi_{m-1}(\bm x(t_0)))}{-L_gL_f^{m-1}b(\bm x(t_0))} < \bm u_{max}$ (or $ > \bm u_{min}$ if $L_gL_f^{m-1}b(\bm x) > 0$), then (\ref{eqn:hocbfcontrol}) is active at $t_0$ and may remain active for all $t\geq t_0$ such that $\bm u$ cannot take the value $\bm u_{max}$. This implies that the system performance is reduced by the HOCBF constraint   (\ref{eqn:hocbfcontrol}).

We want (\ref{eqn:hocbfcontrol}) to be active when $b(\bm x)$ is close to 0. The terms $L_f^{m}b(\bm x)$ and $L_gL_f^{m-1}b(\bm x(t))$ in  (\ref{eqn:hocbfcontrol}) depend only on the safety constraint itself and system (\ref{eq:affine}) (not affected by the definition of HOCBF), and the term $O(b(\bm x))$ (could be positive or negative) depends on the safety constraint, system (\ref{eq:affine}) and the derivatives of class $\mathcal{K}$ functions $\alpha_1,\alpha_2,\dots, \alpha_{m-1}$, while $\alpha_m(\psi_{m-1}(\bm x)) > 0$ depends heavily on these class $\mathcal{K}$ functions and it takes big positive values when $b(\bm x) >> 1, \psi_1(\bm x)>> 1, \dots, \psi_{m-1}(\bm x) >> 1$ such that $\alpha_m(\psi_{m-1}(\bm x)) + O(b(\bm x))$ can also take big positive values. Thus, if these class $\mathcal{K}$ functions are high order polynomial functions, the right hand side of (\ref{eqn:hocbfcontrol}) tends to be bigger (or smaller if $L_gL_f^{m-1}b(\bm x) > 0$) compared with the low order polynomial functions when $b(\bm x) >> 1, \psi_1(\bm x)>> 1, \dots, \psi_{m-1}(\bm x) >> 1$. In other words, the feasible region for $\bm u$ is larger under high order polynomial class $\mathcal{K}$ functions. However, the right hand side of (\ref{eqn:hocbfcontrol}) may be smaller (usually negative) in high order polynomial functions than low order polynomial functions when all of $b(\bm x), \psi_1(\bm x), \dots, \psi_{m-1}(\bm x)$ become small. 

\begin{remark}  \label{rem:feasible}
	The significance of larger feasible region for $\bm u$ lies in the fact that an optimal control problem will not be over-constrained by the time-invariant HOCBF constraint (\ref{eqn:hocbfcontrol}). If a problem is over-constrained, the system performance is reduced. The HOCBF may decrease faster to zero under low order polynomial class $\mathcal{K}$ functions than high order ones since (\ref{eqn:hocbfcontrol}) is less restrictive on $\bm u$ when the HOCBF is close to zero. We will illustrate these properties in the ACC case study in Sec. \ref{sec:ACC}.
\end{remark}

\subsubsection{\textbf{Conflict between Control Input Limitation and HOCBF Constraint} (\ref{eqn:hocbfcontrol})} 
 The constraint (\ref{eqn:hocbfcontrol}) may conflict with $\bm u_{min}$ in (\ref{eqn:controllimit}) (or $\bm u_{max}$ if $L_gL_f^{m-1}b(\bm x) > 0$). If this happens, the optimal control problem becomes infeasible. For the ACC problem defined in   \cite{Aaron2014}, this conflict is addressed by considering the minimum braking distance, which results in another complex safety constraint. 
 
 However, we may need to approximate the minimum braking distance with this method when we have non-linear dynamics and a cooperative optimization control problem \cite{Wei2019}. This conflict is hard to address for high-dimensional systems. Here, we discuss how we may deal with this conflict using the HOCBF introduced in this paper.

When (\ref{eqn:hocbfcontrol}) becomes active, its right hand side should be large enough such that (\ref{eqn:hocbfcontrol}) does not conflict with $\bm u_{min}$. Instead of choosing low order polynomial class $\mathcal{K}$ functions, which conflict with the recommendation from the previous subsection, we add penalties $p_1 > 0, p_2 > 0,\dots, p_m > 0$: 
\begin{equation} \label{eqn:functions_ex}
\begin{aligned}
\psi_1(\bm x) :=& \dot b(\bm x) + p_1\alpha_1(b(\bm x))\\
\psi_2(\bm x) := &\dot \psi_1(\bm x) + p_2\alpha_2(\psi_1(\bm x))\\
&\vdots\\
\psi_m(\bm x) :=& \dot \psi_{m-1}(\bm x) + p_m\alpha_m(\psi_{m-1}(\bm x))
\end{aligned}
\end{equation}

\begin{remark} \label{rem:penalty}
	The penalties $p_1, p_2,\dots, p_m$ also limit the feasible region of $\bm u$ as the class $\mathcal{K}$ functions, but this limitation is weak when $b(\bm x) >> 1, \psi_1(\bm x)>> 1, \dots, \psi_{m-1}(\bm x) >> 1$ such that the idea from the previous subsection may still work to improve the system performance. This is helpful when we want to make the HOCBF constraint (\ref{eqn:hocbfcontrol}) comply with the control limitation by choosing small enough $p_1, p_2,\dots, p_m$ , but the initial conditions should also be satisfied, i.e., $\bm x(t_0) \in C_1(t_0)\cap C_2(t_0)\cap,\dots,\cap C_{m}(t_0)$.
\end{remark}

\section{ACC PROBLEM FORMULATION}
\label{sec:ACC}

In this section, we consider a more realistic version of the adaptive cruise control (ACC) problem introduced in 
Sec.\ref{sec:hocbf:moti}, which was referred to as the simplified adaptive cruise control (SACC) problem. we consider that the safety constraint is critical and study the properties of HOCBF discussed in Sec.\ref{sec:property}.

\subsection{Vehicle Dynamics}

Recall that $S(t)$ denotes the set of vehicle indices in an urban area at time $t$. Instead of using the simple dynamics in (\ref{eqn:simpledynamics}), we consider more accurate vehicle dynamics for $i\in S(t)$ in the form:
\begin{equation}\label{eqn:veh}
m_i\dot v_i(t) = u_i(t) - F_{r}(v_i(t))
\end{equation}
where $u_i(t)$ denotes the control input of vehicle $i$, $m_i$ denotes its mass, $v_i(t)$ denotes its velocity. $F_{r}(v_i(t))$ denotes the resistance force, which is expressed \cite{Khalil2002} as:
\begin{equation}\label{eqn:resistence}
F_{r}(v_i(t)) = f_0sgn(v_i(t)) + f_1v_i(t) + f_2 v_i^2(t),
\end{equation}
where $f_0 > 0, f_1 > 0$ and $f_2 > 0$ are scalars determined empirically. The first term in $F_{r}(v(t))$ denotes the coulomb friction force, the second term denotes the viscous friction force and the last term denotes the aerodynamic drag.

With $\bm x_i(t) := (x_i(t),v_i(t))$, we rewrite the dynamics as:
\begin{equation}\label{eqn:vehicle}
\underbrace{\left[\begin{array}{c} 
	\dot x_i(t)\\
	\dot v_i(t)
	\end{array} \right]}_{ \dot {\bm x_i}(t)}=
\underbrace{\left[\begin{array}{c}  
	v_i(t)\\
	-\frac{1}{m_i}F_{r}(v_i(t))
	\end{array} \right]}_{f(\bm x_i(t))} + 
\underbrace{\left[\begin{array}{c}  
	0\\
	\frac{1}{m_i}
	\end{array} \right]}_{g(\bm x_i(t))}u_i(t)
\end{equation}
where $x_i(t)$ denotes the position in the lane.

$\mathbf{Constraint 1}$ (Vehicle limitations): There are constraints
on the speed and acceleration for each $i\in S(t)$, i.e.,
\begin{equation}\label{eqn:limitation}%
\begin{aligned} v_{min} \leq v_i(t)\leq v_{max}, \forall t\in[t_i^0,t_i^f],\\ -c_dm_ig\leq u_i(t)\leq c_am_ig, \forall t\in[t_i^0,t_i^f], \end{aligned} 
\end{equation}
where $t_{i}^{0},t_{i}^{f}$ denote the time instants that $i_p$ (recall that $i_{p}$ denotes the index of the vehicle which immediately precedes $i$ - if one is
present) precedes $i$ and $i_p$ no longer precedes $i$, respectively. $v_{max}>0$ and $v_{min}\geq 0$ denote the maximum and minimum allowed speeds, while $c_d>0$ and $c_a > 0$ are deceleration and
acceleration coefficients (we use the form (\ref{eqn:limitation}) instead of (\ref{eqn:controllimit}) to note that the maximum and minimum control inputs depend on $m_i$), respectively, and $g$ is the gravity constant.

$\mathbf{Constraint 2}$ (Safety constraint): We require that the distance $z_{i,i_{p}}(t):=x_{i_{p}}(t)-x_{i}(t)$ satisfy
\begin{equation}\label{eqn:safetyACC}%
z_{i,i_{p}}(t)\geq \delta,\text{ \ }\forall t\in\lbrack
t_{i}^{0},t_{i}^{f}], 
\end{equation}
where $\delta>0$ is determined by the length of the two vehicles (generally dependent on $i$ and $i_{p}$ but taken to be a constant over all vehicles for simplicity).

$\mathbf{Objective1}$ (Desired Speed): The vehicle $i\in S(t)$ always attempts to achieve a desired speed $v_d$.

$\mathbf{Objective2}$ (Minimum Energy Consumption): We also want to minimize the energy consumption:
\begin{equation}\label{eqn:energy}
J_{i}(u_{i}(t))=\int_{t_{i}^{0}}^{t_{i}^{f}}\left(\frac{u_i(t) - F_{r}(v_i(t))}{m_i}\right)^2dt,
\end{equation}

\begin{problem}\label{problem1}
	Determine control laws to achieve Objectives 1, 2 subject to Constraints 1, 2, for each vehicle $i\in S(t)$ governed by dynamics (\ref{eqn:vehicle}).
\end{problem}
\vspace{2ex}

We use the HOCBF method to impose Constraints 1 and 2 on control input and a control Lyapunov function \cite{Aaron2012} to achieve Objective 1. We capture Objective 2 in the cost of the optimization problem.

\section{ACC PROBLEM REFORMULATION}
\label{sec:reform}

For {\it\textbf{Problem} 1}, we use the quadratic program (QP) - based method introduced in \cite{Aaron2014}. We  consider three different types of class $\mathcal{K}$ functions (square root, linear and quadratic functions) to define a HOCBF for Constraint 2.

\subsection{Desired Speed (Objective 1)} We use a  control Lyapunov function to stabilize $v_i(t)$ to $v_d$ and relax the corresponding constraint (\ref{eqn:clf}) to make it a soft constraint \cite{Aaron2012}. Consider a Lyapunov function $V_{acc}(\bm x_i(t)):= (v_i(t) - v_d)^2$, with $c_1 = c_2 = 1$ and $c_3 = \epsilon > 0$ in Def. \ref{def:clf}. 
Any control input $u_i(t)$ should satisfy
\begin{equation}\label{eqn:clfconstraint}%
\begin{aligned}
\underbrace{-\frac{2(v_{i}(t) - v_d)}{m_i}F_r(v_i(t))}_{L_fV_{acc}(\bm x_i(t))} + \underbrace{\epsilon(v_i(t) - v_d)^2}_{\epsilon V_{acc}(\bm x_i(t))} \\+ \underbrace{\frac{2(v_{i}(t) - v_d)}{m_i}}_{L_gV_{acc}(\bm x_i(t))}u_i(t) \leq \delta_{acc}(t)
\end{aligned}
\end{equation}
$\forall t\in[t_i^0,t_i^f]$. Here $\delta_{acc}(t)$ denotes a relaxation variable that makes (\ref{eqn:clfconstraint}) a soft constraint.

\subsection{Vehicle Limitations (Constraint 1)} Since the relative degrees of speed limitations are 1, we use HOCBFs with $m = 1$ to map the limitations from speed $v_i(t)$ to control input $u_i(t)$. Let $b_{i,1}(\bm x_i(t)) := v_{max} - v_i(t)$, $b_{i,2}(\bm x_i(t)) := v_i(t) - v_{min}$ and choose $\alpha_1(b_{1,q}) = b_{1,q}, q\in\{1,2\}$ in Def. \ref{def:hocbf} for both HOCBFs. Then any control input $u_i(t)$ should satisfy
\begin{small}
	\begin{equation}\label{eqn:MaxSpeedConstraint}%
	\begin{aligned}
	\underbrace{\frac{F_r(v_i(t))}{m_i}}_{L_fb_{i,1}(\bm x_i(t))}\! +\! \underbrace{\frac{-1}{m_i}}_{L_gb_{i,1}(\bm x_i(t))}u_i(t)\! + \underbrace{v_{max} - v_i(t)}_{b_{i,1}(\bm x_i(t))} \geq\! 0,
	\end{aligned}
	\end{equation}
\end{small}
\begin{small}
	\begin{equation}\label{eqn:MinSpeedConstraint}%
	\begin{aligned}
	\underbrace{\frac{-F_r(v_i(t))}{m_i}}_{L_fb_{i,2}(\bm x_i(t))}\! +\! \underbrace{\frac{1}{m_i}}_{L_gb_{i,2}(\bm x_i(t))}u_i(t)\! + \underbrace{v_i(t) - v_{min}}_{b_{i,2}(\bm x_i(t))} \geq\! 0.
	\end{aligned}
	\end{equation}
\end{small}

Since the control limitations are already constraints on control input, we do not need HOCBFs for them.

\subsection{Safety Constraint (Constraint 2)} The relative degree of the safety constraint (\ref{eqn:safetyACC}) is two. Therefore, we need to define a HOCBF with $m = 2$. Let $b(\bm x_i(t)) := z_{i,i_p}(t) - \delta$. We consider three different forms of class $\mathcal{K}$ functions in Def. \ref{def:hocbf} (with a penalty $p > 0$ on both $\alpha_1, \alpha_2$ for all forms like (\ref{eqn:functions_ex}):

\textbf{Form 1: $\alpha_1$ is linear,  $\alpha_2$ is square root}:
\begin{equation} \label{eqn:sqrt}
\begin{aligned}
\psi_1(\bm x_i(t)) :=& \dot b(\bm x_i(t)) + pb(\bm x_i(t))\\
\psi_2(\bm x_i(t)) := &\dot \psi_1(\bm x_i(t)) + p\sqrt{\psi_1(\bm x_i(t))}
\end{aligned}
\end{equation}

Combining the dynamics (\ref{eqn:vehicle}) with (\ref{eqn:sqrt}), any control input  $u_i(t)$ should satisfy
\begin{equation} \label{eqn:sqrtconstraint}
\begin{aligned}
\underbrace{\frac{F_r(v_i(t))}{m_i}}_{L_f^2b(\bm x_i(t))}\! +\! \underbrace{\frac{-1}{m_i}}_{L_gL_fb(\bm x_i(t))}u_i(t)\! + p\dot b(\bm x_i(t)) \\+  p\sqrt{\dot b(\bm x_i(t)) + pb(\bm x_i(t))}\geq\! 0.
\end{aligned}
\end{equation}

\textbf{Form 2: Both $\alpha_1$ and $\alpha_2$ are linear}:
\begin{equation} \label{eqn:linear}
\begin{aligned}
\psi_1(\bm x_i(t)) :=& \dot b(\bm x_i(t)) + pb(\bm x_i(t))\\
\psi_2(\bm x_i(t)) := &\dot \psi_1(\bm x_i(t)) + p\psi_1(\bm x_i(t))
\end{aligned}
\end{equation}

Combining the dynamics (\ref{eqn:vehicle}) with (\ref{eqn:linear}), any control input  $u_i(t)$ should satisfy
\begin{equation} \label{eqn:linearconstraint}
\begin{aligned}
\underbrace{\frac{F_r(v_i(t))}{m_i}}_{L_f^2b(\bm x_i(t))}\! +\! \underbrace{\frac{-1}{m_i}}_{L_gL_fb(\bm x_i(t))}u_i(t)\! + 2p\dot b(\bm x_i(t)) \\+  p^2b(\bm x_i(t))\geq\! 0.
\end{aligned}
\end{equation}

\textbf{Form 3: Both $\alpha_1$ and $\alpha_2$ are quadratic}:
\begin{equation} \label{eqn:quadratic}
\begin{aligned}
\psi_1(\bm x_i(t)) :=& \dot b(\bm x_i(t)) + pb^2(\bm x_i(t))\\
\psi_2(\bm x_i(t)) := &\dot \psi_1(\bm x_i(t)) + p\psi_1^2(\bm x_i(t))
\end{aligned}
\end{equation}

Combining the dynamics (\ref{eqn:vehicle}) with (\ref{eqn:quadratic}), any control input  $u_i(t)$ should satisfy
\begin{equation} \label{eqn:quadraticconstraint}
\begin{aligned}
\underbrace{\frac{F_r(v_i(t))}{m_i}}_{L_f^2b(\bm x_i(t))}\! +\! \underbrace{\frac{-1}{m_i}}_{L_gL_fb(\bm x_i(t))}u_i(t)\! + 2p\dot b(\bm x_i(t)) b(\bm x_i(t)) \\+  p \dot b^2(\bm x_i(t)) + 2p^2\dot b(\bm x_i(t)) b^2(\bm x_i(t)) + p^3b^4(\bm x_i(t))\geq\! 0.
\end{aligned}
\end{equation}

\subsection{Reformulated ACC Problem} 

We partition the time interval $[t_i^0,t_i^f]$ 
into a set of equal time intervals $\{[t_i^0, t_i^0+\Delta t), [t_i^0 + \Delta t,t_i^0 + 2\Delta t),\dots\}$, where $\Delta t > 0$. In each interval $[t_i^0 +\omega \Delta t, t_i^0 +(\omega+1) \Delta t)$ ($\omega = 0,1,2,\dots$), we assume the control is constant (i.e., the overall control will be piece-wise constant), and reformulate (approximately) {\it\textbf{Problem} 1} as a set of QPs.  Specifically, 
at $t = t_i^0 +\omega \Delta t$ ($\omega = 0,1,2,\dots$), we solve 
\begin{equation} \label{eqn:objACC}
\bm u_i^*(t) = \arg\min_{\bm u_i(t)} \frac{1}{2}\bm u_i(t)^TH\bm u_i(t) + F^T\bm u_i(t)
\end{equation}
\[\begin{small}
\bm u_i(t)\! =\! \left[\begin{array}{c}  
\! u_i(t)\!\\
\!\delta_{acc}(t)\!
\end{array} \right]\!,
H\! =\! \left[\begin{array}{cc} 
\frac{2}{m_i^2} & 0\\
0 & 2p_{acc}
\end{array} \right]\!, F\! =\!  \left[\begin{array}{c} 
\!\frac{-2F_r(v_i(t))}{m_i^2}\!\\
0
\end{array} \right]. \end{small}
\]
subject to
\[
A_{\text{clf}} \bm u_i(t) \leq b_{\text{clf}},
\]
\[
A_{\text{limit}} \bm u_i(t) \leq b_{\text{limit}},
\]
\[
A_{\text{hocbf\_safety}} \bm u_i(t) \leq b_{\text{hocbf\_safety}}, 
\]
where $p_{acc} > 0$ and the constraint parameters are
$$
\begin{aligned}
A_{\text{clf}} &= [L_gV_{acc}(\bm x_i(t)),\qquad -1],\\
b_{\text{clf}} &= -L_fV_{acc}(\bm x_i(t)) - \epsilon V_{acc}(\bm x_i(t)).
\end{aligned}
$$
$$
\begin{aligned}
A_{\text{limit}} &= \left[\begin{array}{cc} 
-L_gb_{i,1}(\bm x_i(t)), & 0\\
-L_gb_{i,2}(\bm x_i(t)), & 0\\
1, & 0
\end{array} \right],\\
b_{\text{limit}} &= \left[\begin{array}{c}  
L_fb_{i,1}(\bm x_i(t)) + b_{i,1}(\bm x_i(t))\\
L_fb_{i,2}(\bm x_i(t)) + b_{i,1}(\bm x_i(t))\\
c_am_ig
\end{array} \right].
\end{aligned}
$$
$$
A_{\text{hocbf\_safety}} = \left[\begin{array}{cc} 
-L_gL_fb(\bm x_i(t)), & 0
\end{array} \right],
$$
$$
b_{\text{hocbf\_safety}}
\left\{
\begin{array}{lcl}
=L_f^2b(\bm x_i(t)) + p\dot b(\bm x_i(t)) \\\;\;\;+  p\sqrt{\dot b(\bm x(t)) + pb(\bm x_i(t))}, \textbf{\quad for Form 1},\\
=L_f^2b(\bm x_i(t)) + 2p\dot b(\bm x_i(t)) \\\;\;\;+  p^2b(\bm x_i(t)), \textbf{\quad for Form 2},\\
=L_f^2b(\bm x_i(t)) + 2p\dot b(\bm x_i(t)) b(\bm x_i(t)) \\\;\;\;+  p \dot b^2(\bm x_i(t)) + 2p^2\dot b(\bm x_i(t)) b^2(\bm x_i(t)) \\\;\;\;+ p^3b^4(\bm x_i(t)), \textbf{\quad for Form 3}.
\end{array}
\right.
$$

 After solving (\ref{eqn:objACC}), we update (\ref{eqn:vehicle}) with $u^*(t)$, $\forall t\in (t_0 +\omega \Delta t, t_0 +(\omega+1) \Delta t)$. 

\begin{remark} 
	The minimum control constraint in Constraint 1 is not included in (\ref{eqn:objACC}). This follows from Remark \ref{rem:penalty} in Sec.\ref{sec:property} that we may choose small enough $p$ such that the minimum control constraint is satisfied.
\end{remark}

\section{IMPLEMENTATION AND RESULTS}
\label{sec:case}
In this section, we present case studies for {\it\textbf{Problem} 1} to illustrate the properties described in Sec.\ref{sec:property}. As noticed in (\ref{eqn:objACC}), the term $L_gL_fb(\bm x_i(t)) = -\frac{1}{m_i}$ depends only on $m_i$. Therefore, the assumption
from the beginning of Sec.  \ref{sec:property} is satisfied.

 All the computations and simulations were conducted in MATLAB. We used quadprog to solve the quadratic programs and ode45 to integrate the dynamics. The simulation parameters are listed in Table \ref{table:param}.

\begin{table}
	\caption{Simulation parameters for {\it\textbf{Problem} 1}}\label{table:param}
	\begin{center}
		\begin{tabular}{|c||c||c|}
			\hline
			Parameter & Value & Units\\
			\hline
			\hline
			$v_{i}(t_i^0)$ & 20& $m/s$\\
			\hline
			$z_{i,i_p}(t_i^0)$ & 100& $m$\\
			\hline
			$\delta$ & 10& $m$\\
			\hline
			$v_{i_p}(t)$ & 13.89& $m/s$\\
			\hline
			$m_i$ & 1650& $kg$\\
			\hline
			g & 9.81& $m/s^2$\\
			\hline
			$f_0$ & 0.1& $N$\\
			\hline
			$f_1$ & 5& $Ns/m$\\
			\hline
			$f_2$ & 0.25& $Ns^2/m$\\
			\hline
			$v_{max}$ & 30& $m/s$\\
			\hline
			$v_{min}$ & 0& $m/s$\\
			\hline
			$\Delta t$ & 0.1& $s$\\
			\hline
			$\epsilon$ & 10& unitless\\
			\hline
			$c_a$ & 0.4& unitless\\
			\hline
			$c_d$ & 0.4& unitless\\
			\hline
			$p_{acc}$ & 1& unitless\\
			\hline
		\end{tabular}
	\end{center}
	
\end{table}

\subsection{\textbf{Case 1: feasible region for control input (comparison between square root, linear (as in \cite{Nguyen2016}) and quadratic function)}} We set $p$ to 1, 1, 0.1 for Forms 1, 2, 3, respectively. In order to study the values of $b(\bm x_i(t))$ for which the HOCBF constraints (\ref{eqn:sqrtconstraint}), (\ref{eqn:linearconstraint}) and (\ref{eqn:quadraticconstraint}) become active, we show how $u_i(t)$ changes as $b(\bm x_i(t))\rightarrow 0$ in Fig.\ref{fig:feasible}. The dashed lines in Fig. \ref{fig:feasible} denote the value of the right-hand side of the HOCBF constraint (like the one in (\ref{eqn:hocbfcontrol})), and the solid lines are the optimal controls obtained by solving (\ref{eqn:objACC}). When the dashed lines and solid lines coincide, the HOCBF constraints (\ref{eqn:sqrtconstraint}), (\ref{eqn:linearconstraint}) and (\ref{eqn:quadraticconstraint}) are active.

\begin{figure}[thpb]
	\centering
	\includegraphics[scale=0.65]{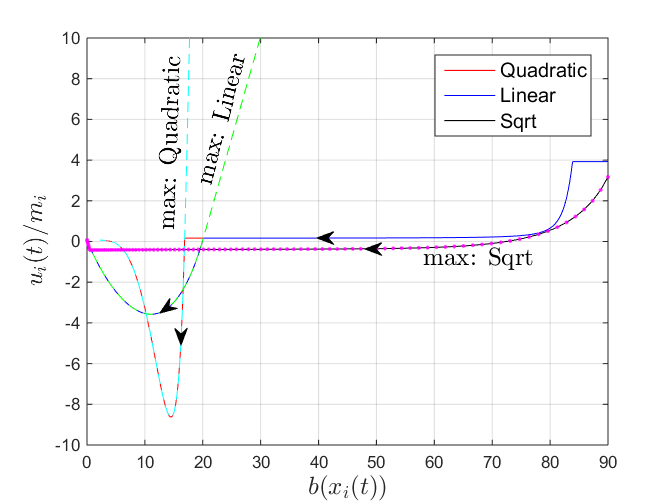}
	\caption{Control input $u_i(t)$ as $b(\bm x_i(t))\rightarrow 0$ for Forms 1, 2, 3 (square root, linear and quadratic class $\mathcal{K}$ functions, respectively). The arrows denote the changing trend for $b(\bm x_i(t))$ with respect to time.}	
	\label{fig:feasible}
\end{figure}

The HOCBF constraint (\ref{eqn:quadraticconstraint}) becomes active when $b(\bm x_i)$ take less value than the other two HOCBF constraints (\ref{eqn:sqrtconstraint}) and (\ref{eqn:linearconstraint}), while the HOCBF constraint (\ref{eqn:sqrtconstraint}) becomes active from the beginning ($b(\bm x_i(t)) = 90$) as shown in Fig.\ref{fig:feasible}. Therefore, the feasible region for $u_i(t)$ is limited when we choose the class $\mathcal{K}$ functions in Form 1, and thus, {\it\textbf{Problem} 1} is over-constrained and the performance of vehicle $i$ is reduced. The feasible region for $u_i(t)$ is bigger under quadratic function than linear function when $b(\bm x_i)$ takes large values but tends to require larger control input $u_i(t)$ after (\ref{eqn:quadraticconstraint}) becomes active. 

\subsection{\textbf{Case 2: conflict between braking limitation and HOCBF constraint}} By Remark \ref{rem:penalty}, we may find a small enough $p$ in (\ref{eqn:linear}) and (\ref{eqn:quadratic}) such that the HOCBF constraints (\ref{eqn:linearconstraint}) and (\ref{eqn:quadraticconstraint}) do not conflict with the minimum control limitation. We present the case studies for the linear and quadratic class $\mathcal{K}$ functions in Fig.\ref{fig:regulate} and Fig.\ref{fig:regulateQ}, respectively.

\begin{figure}[thpb]
	\centering 
	\includegraphics[scale=0.65]{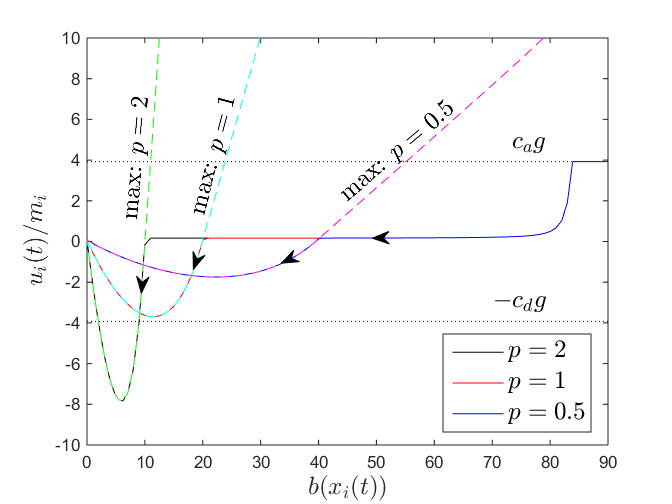}
	\caption{Control input $u_i(t)$ as $b(\bm x_i(t))\rightarrow 0$ for different $p$ values under linear class $\mathcal{K}$ function. The arrows denote the changing trend for $b(\bm x_i(t))$ with respect to time.}
	\label{fig:regulate}	
\end{figure}  

\begin{figure}[thpb]
	\centering
	\includegraphics[scale=0.65]{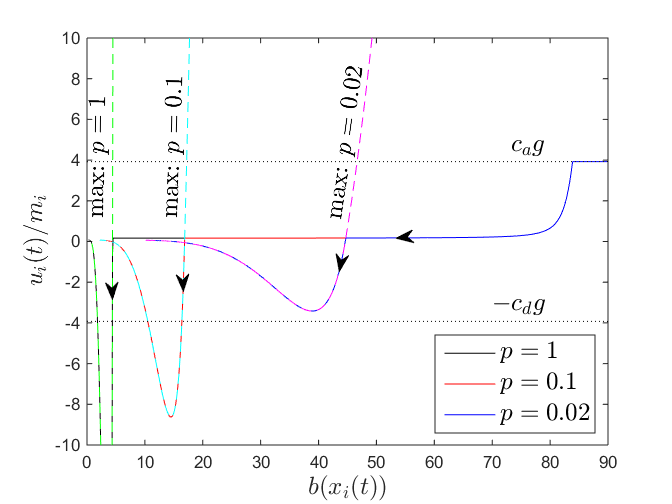}
	\caption{Control input $u_i(t)$ as $b(\bm x_i(t))\rightarrow 0$ for different $p$ values under quadratic class $\mathcal{K}$ function. The arrows denote the changing trend for $b(\bm x_i(t))$ with respect to time.}
	\label{fig:regulateQ}	
\end{figure}

In Fig.\ref{fig:regulate} and Fig.\ref{fig:regulateQ}, the HOCBF constraint does not conflict with the braking limitation when $p = 1$ and $p = 0.02$ for linear and quadratic class $\mathcal{K}$ functions, respectively. The minimum control input increases as $p$ decreases.

Then, we set $p$ to be $ 1, 1, 0.02$ for Forms 1, 2, 3, respectively. We present the speed and control profiles in Fig.\ref{fig:state} and the forward invariance of the set $C_1(t) \cap C_2(t)$, where $C_1(t):= \{\bm x_i(t): b(\bm x_i(t)) \geq 0\}$ and $C_2(t) := \{\bm x_i(t): \psi_1(\bm x_i(t)  \geq 0\}$ in Fig.\ref{fig:sets}.

\begin{figure}[thpb]
	\centering
	\includegraphics[scale=0.65]{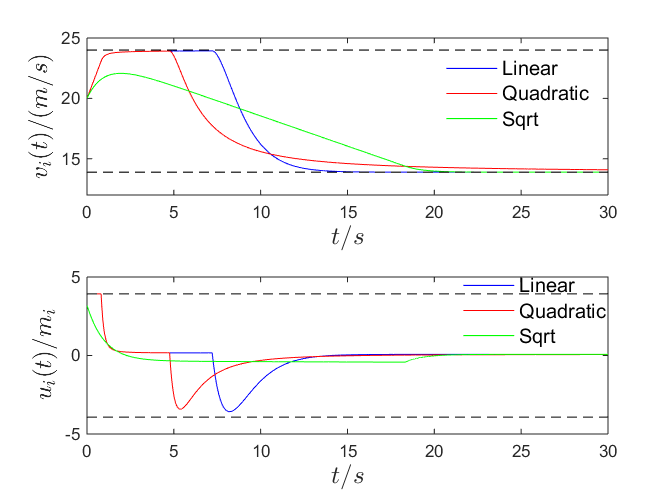}
	\caption{The speed and control profiles for vehicle $i$ under Forms 1, 2, 3 (square root ($p = 1$), linear ($p = 1$) and quadratic ($p = 0.02$) class $\mathcal{K}$ functions, respectively). The ACC problem for the square root class $\mathcal{K}$ function is over-constrained and vehicle $i$ can not reach the desired speed $v_d$.}	
	\label{fig:state}
\end{figure}

\begin{figure}[thpb]
	\centering
	\includegraphics[scale=0.65]{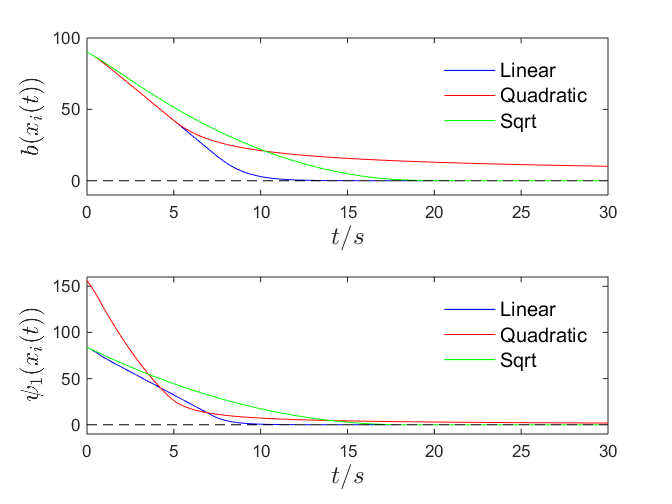}
	\caption{The variation of functions $b(\bm x_i(t))$ and $\psi_1(\bm x_i(t)$ under Forms 1, 2, 3 (square root ($p = 1$), linear ($p = 1$) and quadratic ($p = 0.02$) class $\mathcal{K}$ functions, respectively). $b(\bm x_i(t))\geq 0$ and $\psi_1(\bm x_i(t)\geq 0$ imply the forward invariance of the set $C_1(t)\cap C_2(t)$.}	
	\label{fig:sets}
\end{figure}

We can increase the $p$ value for Form 1 such that the HOCBF constraint (\ref{eqn:sqrtconstraint}) is not over-constrained. For $p = 2$, we show the speed and control profiles in Fig.\ref{fig:state_n} and the forward invariance of the intersection set in Fig.\ref{fig:sets_n}. The HOCBF $b(\bm x_i(t))$ decreases faster to 0 for the square root class $\mathcal{K}$ function than the linear class $\mathcal{K}$ function and tends to stay away from 0 under quadratic class $\mathcal{K}$ function. As shown in Fig.\ref{fig:sets_n}, the values of HOCBF $b(\bm x_i(t)) $ are $ 0.0193, 0.0413, 15.6669$ at $t = 15s$ and  $2.8964\times 10^{-7}, 4.4685\times 10^{-4}, 12.9729$ at $t = 20s$ for square root, linear and quadratic class $\mathcal{K}$ functions, respectively.

\begin{figure}[thpb]
	\centering
	\includegraphics[scale=0.65]{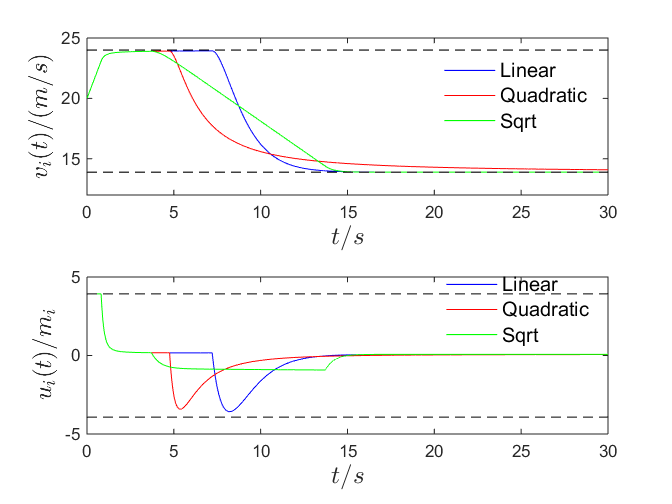}
	\caption{The speed and control profiles for vehicle $i$ under Forms 1, 2, 3 (square root ($p = 1$), linear ($p = 1$) and quadratic ($p = 0.02$) class $\mathcal{K}$ functions, respectively). The solutions to the ACC problem behave well for all three class $\mathcal{K}$ functions.}	
	\label{fig:state_n}
\end{figure}

\begin{figure}[thpb]
	\centering
	\includegraphics[scale=0.65]{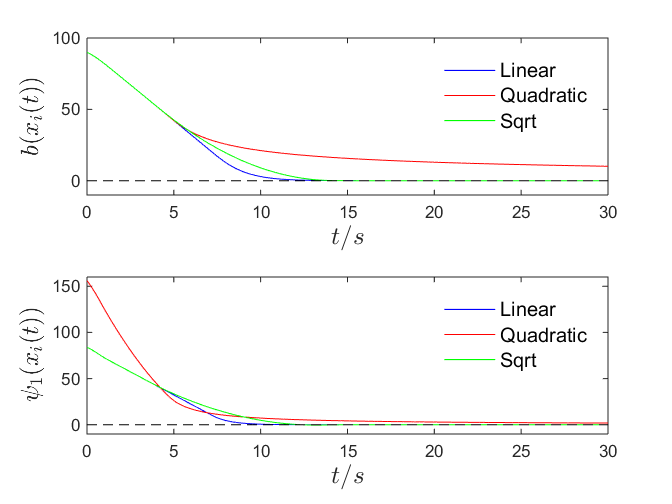}
	\caption{The variation of functions $b(\bm x_i(t))$ and $\psi_1(\bm x_i(t)$ under Forms 1, 2, 3 (square root ($p = 1$), linear ($p = 1$) and quadratic ($p = 0.02$) class $\mathcal{K}$ functions, respectively). $b(\bm x_i(t))\geq 0$ and $\psi_1(\bm x_i(t)\geq 0$ imply the forward invariance of the set $C_1(t)\cap C_2(t)$.}	
	\label{fig:sets_n}
\end{figure}

\section{CONCLUSION \& FUTURE WORK}
\label{sec:conclusion}

We presented an extension of control barrier functions to high order control barrier functions, which allows to deal with high relative degree systems. We also showed how we may deal with the conflict between the HOCBF constraints and the control limitations. We validated the approach by applying it to an automatic cruise control problem with 
constant safety constraint. In the future, we will apply the HOCBF method to more complex problems, such as differential flatness in high relative degree system and bipedal walking.

\addtolength{\textheight}{-6cm}   





\bibliographystyle{plain}
\bibliography{HOCBF}

\end{document}